\documentclass[runningheads]{llncs}

\usepackage{amsmath}
\usepackage{amssymb}
\usepackage{enumitem}

\DeclareMathAlphabet\rsfscr{U}{rsfso}{m}{n}

\def \Z 	{\mathbb{Z}}
\def \B		{\mathbb{B}}
\def \R 	{\mathbb{R}}
\def \F 	{\mathbb{F}}
\def \A 	{\mathbf{A}}
\def \L 	{\mathcal{L}}

\def \CA 	{\mathsf{CA}}
\def \LCCA 	{\mathsf{LCHCA}_{n,p}^M}
\def \neighbor  {\rsfscr{N}}
\def \D 	{\rsfscr{D}}
\def \d 	{\rsfscr{D}'}

\def \SIS	{\mathsf{SIS}_{n,m}^{q,\beta}}
\def \SRPS	{\mathsf{SLS}_{n,m}^{q,\beta}}
\def \DLP	{\mathsf{DLP}_{n,p}}
\def \SDP	{\mathsf{SDDP}_{k,n}^{\delta}}
\def \FDP	{\mathsf{FDP}_{k,n}}
\def \DDP	{\mathsf{DDP}_{n,p}^M}

\begin{document}
\title{Generating Hard Problems of Cellular Automata}

\titlerunning{Hard Problems of Cellular Automata}

% If the paper title is too long for the running head, you can set
% an abbreviated paper title here
% 
\author{Souvik Sur}

% \authorrunning{XXXX et al.}
% First names are abbreviated in the running head.
% If there are more than two authors, 'et al.' is used.

\institute{Department of Computer Science and Engineering,\\ 
Indian Institute of Technology Kharagpur,\\ Kharagpur, West Bengal, India \\
\email{souviksur@iitkgp.ac.in}}

\maketitle
\begin{abstract}
% \hyphenpenalty=2000\exhyphenpenalty=2000\tolerance=500\pretolerance=500%
We propose two hard problems in cellular automata. In particular the problems are,
\begin{itemize}
 \item \noindent  [$\DDP$] Given two \emph{randomly} chosen configurations $t$ and $s$ 
of a cellular automata of length $n$, find the number of transitions $\tau$ between $s$ and $t$.

\item \noindent [$\SDP$] Given two \emph{randomly} chosen configurations $s$ 
of a cellular automata of length $n$ and $x$ of length $k<n$, 
find the configuration $t$ such that $k$ number of cells of $t$ is fixed to $x$ and 
$t$ is reachable from $s$ within $\delta$ transitions.
\end{itemize}

We show that the discrete logarithm problem over the finite field reduces to $\DDP$ and 
the short integer solution problem over lattices reduces to $\SDP$ . The advantage of using 
such problems as the hardness assumptions in cryptographic protocols is that proving the security 
of the protocols requires only the reduction from these problems to the designed protocols. We 
design one such protocol namely a proof-of-work out of $\SDP$.

\end{abstract}


\section{Introduction}\label{introduction}
Cellular automata $\CA$ is a universal model of computation~\cite{Worsch1993Parallel} like Turing machines.
It has been used in numerous applications ranging from cryptography to coding theory, 
from VLSI design to memory testing.
In most of the applications, it succeeds 
to achieve some predefined set of desired properties like pseudo-randomness, efficient parallelizability etc.
Unfortunately, in the context of cryptography nowadays $\CA$s are often referred to ``older crypto", due to 
the absence of theorems to prove certain security properties required for any cryptographic schemes. 
The best that $\CA$s offer for cryptographic design are some conjectures like given a configuration of a $\CA$ it 
is infeasible to find one of its predecessors.

In this paper, we show that (at least) a type of $\CA$, namely linear cyclic hybrid cellular automata 
$\LCCA$ essentially simulates the computation over the finite field $\F_{p^n}$. 
Linear hybrid cellular automata are finite state machines that mimics a linear transformation 
over the vector space defined over a finite field.
The mapping between $\LCCA$ and finite field $\F_{p^n}$ immediately gives us an edge to convert those 
longstanding conjectures into theorems. 

More importantly as a consequence of this mapping, we 
pose two new problems, namely Discrete Distance Problem $\DDP$ and Short Discrete Distance Problem $\SDP$,
over $\LCCA$. We give two polynomial time reductions,
\begin{itemize}
 \item the Discrete Logarithm Problem $\DLP$ over the finite field $\F_{p^n}$ to $\DDP$
 \item the Short Integer Solution Problem $\SIS$ over the lattices to $\SDP$.
\end{itemize}
This reductions shows that cryptographic protocols based on $\DDP$ and $\SDP$ are secure 
as long as $\DLP$ and $\SIS$ computationally hard.
Moreover $\SIS$ known to be an average-case 
hard problem~\cite{Ajtai1996SIS}. As $\SIS$ reduces to $\SDP$ for randomly chosen instance of $\SDP$,
it turns out to be hard in average-case too.
As a typical application we design a proof-of-work scheme out of the problem $\SDP$ at the end.
 
\section{Preliminaries}\label{preliminaries}
In this section we take a brief review of the tools and techniques that will be required 
for the rest of the paper.
\subsection{Rings and Finite Fields}
We take all the definitions related to abstract algebra from this book~\cite{Gallian2016Algebra}.
A ring is an Abelian group under addition, also having an associative multiplication 
that is left and right distributive over addition. The rings that allow the multiplication 
to be commutative is called commutative rings. Suppose $R$ is a commutative ring. A ring of polynomials 
$R[x]=\{\sum\limits_{i=0}^{n-1} a_ix^i\}$ where $a_i\in R$ and $n\in \mathbb{Z}^+$, is the set of 
polynomials whose coefficients are from the ring $R$.
A field $\mathbb{F}$ is a commutative ring with unity in which every nonzero element is a unit. 
As $\mathbb{F}$ is a commutative ring, the set $\mathbb{F}[x]$ is also a ring of polynomials.
If $\mathbb{F}_q$ has a finite number $q$ of elements we call $\mathbb{F}_q$ as a finite field 
of order $q$. It can be shown that $\mathbb{F}_q$ to be a finite field if and only if 
$q=p^n$ where $p$ and $n \in\mathbb{Z}^+$. Every fields of the same order are isomorphic. 
The nonzero elements of a finite field $\mathbb{F}_q$ form a cyclic multiplicative group $\mathbb{F}_q^*$. 
Suppose $\alpha$ is one of the generators $\mathbb{F}_q^*$ then $\alpha$ is called a primitive element 
of the field $\mathbb{F}_q$. We denote $\mathbb{F}_q\langle\alpha\rangle=\{0,1,\alpha,\alpha^2\ldots\alpha^{q-2}\}$ 
to be the field generated by the primitive element $\alpha$.

We call a subset $\mathbb{F}_d\subset\mathbb{F}_q$ is a subfield of $\mathbb{F}_q$ if $\mathbb{F}_d$ preserves 
all the operations of $\mathbb{F}_q$. It can be shown that $d$ always divides $n$. 
% In other words $d=p^m$ for some $m<n$. 
We call $\mathbb{F}_p$ as the base field and $\mathbb{F}_d$ and $\mathbb{F}_q$ as the extended 
fields. We denote these field extensions as $\mathbb{F}_p/\mathbb{F}_d$ and $\mathbb{F}_d/\mathbb{F}_q$.

An ideal is a subset $I$ of elements in a ring $R$ that forms an additive group such that, 
for $x\in R$ and $y\in I$, then $xy\in I$ and $yx \in I$. We call $R/I$ as a quotient ring 
if $R/I$ is the set of cosets of $I$ in $R$ with respect to addition and multiplication.
An ideal $I$ of a ring $R$ is maximal 
if and only if there is no other ideal in between $I$ and $R$. An ideal $\langle a \rangle=\{ra|r\in R\}$
is called a principle ideal. The quotient ring $R/I$ is a field if and only if $I$ is a maximal ideal.

Suppose $f(x)\in\mathbb{F}_p[x]$ is a polynomial over the field $\mathbb{F}_p$. We call $f(x)$ to be 
an irreducible polynomial if and only if $f(x)$ can not be factored into non-constant polynomials
over the field $\mathbb{F}_p$. As $\F_p[x]$ is a ring $\langle f(x)\rangle$ is an ideal in 
$\mathbb{F}_p[x]$. The ideal $\langle f(x) \rangle$ is maximal if and only if $f(x)$ is irreducible 
over $\mathbb{F}_p$. Therefore, the quotient ring $\mathbb{F}_p[x]/\langle f(x) \rangle$ 
is a field if and only if $f(x)$ is irreducible.

\subsection{Discrete Logarithm Problem}

\begin{definition}{\bf (Discrete Logarithm Problem \textsf{DLP}).}
Given a finite cyclic multiplicatively written group $\mathbb{G}=\langle g \rangle$,
a generator $g$ of $\mathbb{G}$, and an element $h\in\mathbb{G}$,
the $\DLP$ is to find an integer $x$, unique modulo the order of
$\mathbb{G}$, such that $h=g^x$. 
\end{definition}
For certain groups $\mathbb{G}$, the $\DLP$ turns out to be a computationally intensive problem.
With the relevance to our current context we mention such a group. 
\noindent
\begin{definition} {\normalfont  \textbf {(DLP over Finite Fields $\DLP$).}}\quad
Let $g$ be a generator of $\mathbb{F}_{p^n}^*$, and $a\in\mathbb{F}_{p^n}^*$
for some prime $p$. There is an integer $x$, unique modulo~$p^n-1$, such
that $a=g^x\in\mathbb{F}_{p^n}$. The $\DLP$ is the problem of determining $x$ from the pair $(g,a)$.
\end{definition}

\subsection{Lattices}
An $n$-dimensional lattice is the set of all integer combinations 
$\{\sum \limits_{i=1}^{n} x_i \beta_i \mid x_i \in \Z\}$ of $n$ linearly independent 
vectors $\{\beta_1, \ldots, \beta_n\}$ in $\R^n$. The set of vectors $\{\beta_1, \ldots , \beta_n\}$ is called a basis 
for the lattice, and can be compactly represented by the matrix 
$\B = [\beta_1 |, \ldots, | \beta_n] \in \R^{n\times n}$ having the basis vectors as columns. 
The lattice generated by $\B$ is denoted by $\L(\B)$. 
% For any basis $\B$, we define the fundamental parallelepiped $P(B) = {Bx: ∀i.0 ≤xi < 1}$.

The minimum distance of a lattice $\L(\B)$ is the minimum distance between any two (distinct) lattice
points and equals the length of the shortest nonzero lattice vector. The minimum distance can be defined
with respect to any norm. For any $p \ge 1$, the $p$-norm of a vector $\vec{x}$ is defined by 
$\|\vec{x}\|_p=\sqrt[p]{\sum_i |z_i|^p}$ corresponding minimum distance is denoted,

$$\lambda^p_1(\L(\B))=min\{\|\vec{x}-\vec{y}\|_p \mid \vec{x}\ne \vec{y} \in \L(\B)\}=
min\{\|\vec{x}\|_p \mid \vec{x} \in \L(\B)\setminus\vec{0}\}.$$

Without loss of generality the $\ell_2$-norm is used in rest of the paper.
\begin{definition}{ \normalfont \textbf{ ($\ell_2$-norm).}} 
The $\ell_2$-norm $\|(z_1, z_2, \ldots, z_n)\|=\sqrt{z_1^2 + z_2^2 + \ldots + z_n^2}$.
\end{definition}
\begin{definition}{ \normalfont \textbf{ ($\gamma$-approximate Shortest Vector Problem $\mathsf{SVP}^\gamma$).}} 
Given a lattice $\L(\B)$, find a nonzero vector $\vec{x} \in \L(\B)$ such that 
$\|\vec{x}\| \le \gamma \lambda_1(\L(\B)).$
\end{definition}

$\mathsf{SVP}^{\gamma}$ is known to be hard for $\gamma=\mathrm{poly}(n)$.

\begin{definition}{\normalfont \textbf{((Inhomogeneous) Short Integer Solution Problem $\SIS$).}}
Given $m$ uniformly random vectors $\vec{a}_i\in_R \Z^n_q$, 
forming the columns of a matrix $\mathbf{A} \in \Z^{n\times m}_q$, 
and a uniformly random vector $\vec{v}\in_R  \Z^n_q$,
find a nonzero integer vector $\vec{z} \in \Z_q^m$ such that,
\begin{enumerate}
 \item the $\ell_2$-norm $\|\vec{z}\| \leq \beta$ ,
 \item $\mathbf{A}\vec{z}=\sum \vec{a}_i \cdot z_i=\vec{v}\in\Z^n_q$.
\end{enumerate}
\end{definition}

The above-mentioned definition of $\SIS$ known to be the inhomogeneous version of $\SIS$.
Historically, $\SIS$ made its debut with its homogeneous version fixing $\vec{v}=\vec{0}\in\Z^n_q$~\cite{Ajtai1996SIS}.
It is shown that $\mathsf{SVP}^{\mathrm{poly}(n)}$ in the worst case reduces to $\SIS$ even if $\A$ 
is chosen uniformly at random i.e., in the average case.
Both these versions are equally hard for typical parameters
for $m \geq \lceil n\log q \rceil$ and $q \gg \beta$, however $\beta > \sqrt{n\log q}$ is required to guarantee that 
there exists at least one solution.

\subsection{Cellular Automata}\label{sec: CA}

Cellular automata is a universal model of computation like Turing machine.

\begin{definition}{\normalfont \textbf{(Cellular Automata).}}\label{CA}
A cellular automata $\CA$, having $n$ cells and $\neighbor$-neighborhoodness,
is a tuple $\mathsf{CA} = \langle Q,n, s, F, \neighbor , \delta\rangle$ with the following meaning,
\begin{enumerate}
 \item $Q$ is the finite and nonempty set of states and also tape alphabet.
 \item $n$ is the length of the $\CA$ i.e., the number of cells. 
%  Conventionally, the cells are numbered from left to right. 
 \item $s\in Q^n$ is the initial configuration\footnote{Configuration is often referred as state of $\CA$. We 
 prefer to use state as the state of a control unit not of the entire $\CA$.}.
 \item $F$ is the set of halting states.
 \item Given the index $i$ of a cell, the neighborhood $\neighbor$ is the set of relative offsets 
 from the $i$-th cell, from which the local transition function assume inputs i.e, $f_i:Q^\neighbor\rightarrow Q$. 
%  Note that $h-1=|\neighbor|$ is the neighborhoodness. One head is designated to write in the respective cells.
 \item $\delta : Q^n \rightarrow Q^n$ is the global transition function which is an ensemble of 
 $n$ number of local transition functions $\{f_i : Q^\neighbor \rightarrow Q\}$.
 \end{enumerate}
\end{definition}

% 
% Each control unit (CU) has a fixed number of input heads that reads inputs from a fixed 
% set of tape cells known as its neighborhood $\neighbor$. 
 Given the index $i$ of a cell we denote its neighborhood set 
 $\neighbor_i\subseteq \{-i,-i+1,\ldots,-1,0,+1,\ldots, n-i-2, n-i-1\}$ 
 where $0$ corresponds to the $i$-th cell itself, starting from the leftmost cell. When the 
 neighborhoods $\neighbor_i$ are regular for all the cells, we drop the subscript $i$
 and the cardinality of the set $\neighbor$ is called the neighborhoodness of the $\CA$. 
 For example, the three-neighborhoodness $\neighbor = \{ -1,0, +1\}$ denotes the neighborhood of 
 $\{ i-1,i, i+1\}$ for any $1<i<n$. 
% Throughout this paper we will ith only $\Sigma=Q$ and will omit $\Sigma$ from the tuple $\CA$.
% Suppose $\rsfscr{C}$ is a $\CA$ of length $n$. 

Any configuration $t=\{f_i(\neighbor_i)\}\in Q^n$ is an ensemble of the images 
of $f_i$ on the states of $\neighbor_i$ . 
% In order to allow each cell $i$ to assume inputs from its neighborhood we must have $c_i=f_i(\neighbor_i)$.
A unit of computations is considered to be a single transition $t'\leftarrow \delta(t)$.
Thus a $\CA$, initialized with a global configuration $s$, makes its transitions through the 
sequence of configurations $\{s, \delta(s),\delta^2(s), \ldots,\delta^{|Q|^n-1}(s)\}$.
Given a $\CA$ there may exists set of global configurations $\Delta\subset Q^n$ such that $\delta(s)=s$ if $s\in\Delta$.
These configurations $\Delta$ are called as dead configurations.

\begin{theorem}\label{thm:independence}
If the initial configuration $s$ of a cellular automata $\mathsf{CA} = \langle Q, n, s, F, \neighbor , \delta\rangle$ 
is chosen uniformly at random i.e, $s\in_R Q^n$, then for any 
$0<\tau<|Q|^n$, 
all the cells of the subsequent configuration $\delta^\tau (s)$ remain independent of 
one another. 
\end{theorem}
 
\begin{proof}
We proceed by induction on $\tau$. For $\tau=0$, the result follows from the
initialization hypothesis as $s\in_R Q^n$. For the inductive step, suppose that all the cells of the 
state $t=\delta^\tau (s)$ are independent at some time $\tau\ge0$ and $t'=\delta(t)=\delta^{\tau+1}(s)$. 
Assuming $v[i]$ denotes the $i$-th cell of the configuration $v$, it suffices to show that 
\begin{equation}\label{eq:independence}
Pr[t'[i] = a \mid  t'[j] = b ] = Pr[t'[i] = a]
\end{equation}
for all $i,j$ with $i\ne j$, and for all $a,b\in Q$.

We denote $\neighbor_i$ as the neighborhood of the $i$-th cell irrespective of the configuration. 
Trivially Eq.~\ref{eq:independence} holds when 
If 
$\neighbor_i\cap\neighbor_j=\{\varnothing\}$. 
So assume that $\neighbor_i\cap\neighbor_j\ne\{\varnothing\}$, we have 
\begin{align*}
t'[i]=f_i(\{\neighbor_i\cap\neighbor_j\}\cup\{\neighbor_i\setminus\neighbor_j\}) \\
t'[j]=f_j(\{\neighbor_i\cap\neighbor_j\}\cup\{\neighbor_j\setminus\neighbor_i\})
\end{align*}
The cell $t'[i]$ depends on $\{\neighbor_i\setminus\neighbor_j\}$ but not on $\{\neighbor_j\setminus\neighbor_i\}$. 
Similarly, the bit $t'[j]$ depends on $\{\neighbor_j\setminus\neighbor_i\}$ but not on $\{\neighbor_i\setminus\neighbor_j\}$.
Evidently, $\{\neighbor_i\setminus\neighbor_j\}\ne\{\neighbor_j\setminus\neighbor_i\}\ne\varnothing$.
By induction hypothesis, all the cells of the configuration $t$ are independent of each other, 
so $\{\neighbor_i\setminus\neighbor_j\}$ and $\{\neighbor_j\setminus\neighbor_i\}$ are independent of each other.
Therefore Eq.~(\ref{eq:independence}) holds.
\end{proof}

\begin{definition}{\normalfont \textbf{(Cyclic Cellular Automata).}}\label{CCA}
A $\mathsf{CA} = \langle Q, n, s, F, \neighbor , \delta\rangle$ with $\Delta$ dead states 
is a cyclic $\CA$ if and only if its initial configuration $s=\delta^k(s)$ for some 
$k \leq |Q|^n-1-|\Delta|$ and $s\ne\delta^i(s)$ for all $i<k$. 
\end{definition}
Note that $k$ is an invariant of the initial configuration $s$ but depends only on $\delta$.
Essentially the transition graph of these $\CA$ looks like a cycle, however, 
there may be multiple such disjoint cycles. Based upon the length of this transition cycle 
we may classify $\CA$s into these two categories.

\begin{definition}{\normalfont \textbf{(Maximum-Length Vs. Group Cellular Automata).}}
A cyclic cellular automata $\mathsf{CA} = \langle Q, n, s, F, \neighbor , \delta\rangle$
with $\Delta$ dead states is a Maximum-Length $\CA$ if and only if 
its initial configuration $s=\delta^k(s)$ for $k=|Q|^n-1-|\Delta|$ and 
$s\ne\delta^i(s)$ for all $i<k$. On the other hand, it is a group $\CA$ if its initial configuration 
$s=\delta^k(s)$ for some $k<|Q|^n-1-|\Delta|$ and 
$s\ne\delta^i(s)$ for all $i<k$. 
\end{definition}

Here, the term ``Maximum-Length" corresponds to the length of the cycle in the transition graph 
of the $\CA$. In Sect~\ref{relation}, we will see that for group $\CA$s, $k\mid(|Q|^n-1-|\Delta|)$
i.e, all the disjoint cycles have equal lengths.

% In terms of power to recognize languages $\CA$ is an universal model of computations. 
% In particular, Theorem. 12 in \cite{Worsch1993Parallel} shows 
% that the computation done by a $\mathsf{PTM}$ with $\mathcal{O}(n)$ CUs using $\mathcal{O}(n)$ cells in its tape
% in $\mathcal{O}(T)$ time is equivalent to that done by a $\CA$ of length $\mathcal{O}(n)$ in time $\mathcal{O}(T)$.

Like Turing machines, $\CA$s are also required to be encoded.
The easiest way to do it is to use a bijection $Q\rightarrow \Z_{|Q|}$ so that the 
transition functions $\delta:\Z_{|Q|}^n\rightarrow \Z_{|Q|}^n$ can be defined mathematically over the 
vector space $\Z_{|Q|}^n$. In some cases, $\delta$ results into functions having algebraic closed forms,
however, may not be possible always. Depending upon the algebraic closed form (if available) of 
$\delta$, $\CA$s can be characterized as,

\begin{definition}{\normalfont \textbf{(Linear Cellular Automata).}}
A cellular automata $\mathsf{CA} = \langle \Z_{|Q|}, n, s, F, \neighbor , \delta\rangle$ is linear
if and only if its transition function $\delta$ can be represented 
by a linear operator $g:\Z_{|Q|}^n\rightarrow \Z_{|Q|}^n$. 
\end{definition}

It means for any $\vec{v}_i\in\Z_{|Q|}^n$ and any $k_i\in\Z_{|Q|}$
we have $g(\sum_{i} k_i.\vec{v}_i)=\sum_i k_i.g(\vec{v}_i)$.

\begin{description}
 \item [Transition Matrix $M$ of a linear $\CA$] 
 As $g$ is a linear operator over the vector space $\Z_{|Q|}$, there is a matrix $M\in\Z_{|Q|}^{n\times n}$ 
such that for any $\vec{v}\in \Z_{|Q|}^n$, we have $g(\vec{v})=M\times\vec{v}$. Therefore, given an initial configuration 
$\vec{s}\in \Z_{|Q|}^n$ of a linear $\CA$, we have the sequence of configurations $\{\vec{s}, M\vec{s}, M^2\vec{s},\ldots, 
M^{|Q|^n-1}\vec{s}\}$. Traditionally $M$ is called as the transition matrix of the linear $\CA$.
The $\delta$ of a linear $\CA$ can be represented with its transition matrix $M\in\Z^{n \times n}_{|Q|}$. 
%  Given $\delta$ as an ensemble of local 
%  transition functions $\{f_i\}$, the construction of $M$ is efficient. The key observation is that $\delta$ itself is a
% linear operator in the standard basis of the vector space $\Z_{|Q|}^n$.
%  Suppose $\vec{t}=M\vec{s}$ and $\vec{v}[i]$ denotes the $i$-th coordinate of any vector $\vec{v}$. 
%  Then express $\vec{t}[i]=f_i{(\neighbor_i)}=\sum\limits_{j=1}^{n} m_{ij}.\vec{s}[j]$ for some $m_{ij}\in \Z_{|Q|}$. 
%  Now the transition matrix $M=\{m_{ij}\}\in\Z^{n \times n}_{|Q|}$. 
 For any linear $\CA$, $\vec{s}=\delta(\vec{s})$ if and only if $\vec{s}=\vec{0}$, 
 so $\Delta=\{\vec{0}\}$.
\end{description}

\begin{theorem}\label{thm:unbiased}
If the initial configuration of a linear cellular automata 
$\mathsf{CA} = \langle Q, n, \vec{s}, F, \neighbor , M\rangle$ 
is chosen uniformly at random i.e, $\vec{s}\in_R Z^n_{|Q|}$, then for any $0<\tau<|Q|^n$, all the coordinates of 
the subsequent configuration $\vec{t}=M^\tau \vec{s}$ remain uniformly unbiased. 

% with probability $$\Pr[t[i]=a\in Z_{|Q|}]=\frac{1}{|Q|}$$ for all $a\in Z_{|Q|}$.
\end{theorem}

\begin{proof}
We proceed by induction on $\tau$. For $\tau=0$, the result follows from the
initialization hypothesis as $\vec{s}\in_R Z^n_{|Q|}$. For the inductive step, suppose that all the coordinates of the 
configuration $\vec{t}=M^\tau \vec{s}$ are unbiased at some time $\tau\ge0$ and $\vec{t}'=M\vec{t}=M^{\tau+1}\vec{s}$. 
Assuming $\vec{v}[i]$ denotes the $i$-th coordinate of the configuration $\vec{v}$, it suffices to show that,
\begin{equation}\label{eq:unbiased}
Pr[t'[i] = a \mid  t[i] = b ] = Pr[t'[j] = a]=\frac{1}{|Q|}.
\end{equation}
for all $i,j$ including $i= j$, and for all $a,b\in \Z_{|Q|}$.

Suppose $M=\{m_{ij}\}$ then $t'[i]=\sum \limits_{j=1}^{n}m_{ij}t[j]$. 
Observe that addition and multiplication over the set $Z_{|Q|}$ are unbiased 
because of the modulo reduction. For each of the operations, there are exactly $|Q|^2$ number of pairs 
which are mapped to $|Q|$ number of elements through modulo reduction. Therefore each of these operations 
is unbiased with probability $\frac{|Q|}{|Q|^2}=\frac{1}{|Q|}$.
Further by the induction hypothesis, 
$t[j]$s are unbiased. As a result $t'[i]=\sum \limits_{j=1}^{n}m_{ij}t[j]$ becomes unbiased. It does not matter 
if the non-zero $m_{ij}$s are biased as they are fixed already.

By Theorem~\ref{thm:independence}, all the coordinates are independent of each other.
Putting these two together, for all $i,j$ including $i= j$, and for all $a,b\in \Z_{|Q|}$
for Eq.~\ref{eq:unbiased} holds true.
\end{proof}

Theorem~\ref{thm:unbiased} implies that there is no
leakage of information from one configuration to the next one. One must evaluate all of the coordinates 
to obtain the next configuration from the present configuration.
Theorem~\ref{thm:independence} and Theorem~\ref{thm:unbiased} together indicate
that this sequence of configurations $\{\vec{s}, M\vec{s}, M^2\vec{s},\ldots, M^{|Q|^n-1}\vec{s}\}$ for a 
linear $\CA$ looks like a pseudo-random sequence of vectors chosen uniformly at random from the vector space $Z^n_{|Q|}$ 
with a periodicity of $(|Q|^n-1)$ as $\Delta=\{\vec{0}\}$.

For linear $\CA$s, it is not necessary to have a common linear form for 
all the local transition functions $f_i:Q^\neighbor \rightarrow Q$. For example, an 
$f_i$ may ignore one of its neighbors $k\in\neighbor$ by making $m_{i,i+k}=0$ while another $f_j$ includes 
the same offset $k\in\neighbor$ keeping $m_{j,j+k}\ne 0$.

\begin{definition}{\normalfont \textbf{(Hybrid Vs. Uniform Cellular Automata).}}
 A linear cellular automata $\mathsf{CA} = \langle Q, n, s, F, \neighbor , M\rangle$
 is a uniform $\CA$ if and only if given a row $i$ of the transition matrix $M$ \emph{any} 
 other row $j>i$ can be obtained by $j-i$ number of right shifts of the $i$-th. A $\CA$ is hybrid if it is not uniform.
\end{definition}

\begin{description}
 \item [The Characteristic Polynomial of $M$]
 The characteristic polynomial $f_M(x)$ of the transition matrix helps us to identify the uniformity of 
a linear $\CA$. If $f_M(x)$ is reducible over the set $Z_{|Q|}$ then the $\CA$ is a uniform one~\cite{Cattell1996Analysis},
else if $f_M(x)$ is irreducible over the set $Z_{|Q|}$ then $M$ generates a hybrid $\CA$. Additionally, if 
$f_M(x)$ is a primitive (also irreducible) polynomial over the set $Z_{|Q|}$ then $M$ generates a maximum-length hybrid $\CA$.
\end{description}

\section{Hard Problems on Cellular Automata}\label{prob}
In this section, we will define two computational problems based on the 
transition of configurations of a particular type of linear cyclic cellular automata.
As $\Z_{|Q|}$ becomes a field when $|Q|=p$ is a prime, we define our final characterization 
of $\CA$s as follows,

\begin{definition}{\normalfont \textbf{(Linear Cyclic Hybrid Cellular Automata over $\F_p$ $\LCCA$).}}
A cellular automata $\LCCA = \langle \F_p, n, \vec{s}, F, \neighbor , M\rangle$ is linear cyclic hybrid cellular automata
over the field $\F_p$, if and only if it is linear, cyclic, hybrid and $Q=\F_p$ for some prime $p$.
\end{definition}

\subsection{Discrete Distance Problem}
We observe that the exponent $\tau$ of $M$ in the expression $\vec{t}=M^\tau\vec{s}$ acts as an 
discrete distance between two vectors $\vec{t}$ and $\vec{s}$.
So we name this problem as \textbf{D}iscrete \textbf{D}istance \textbf{P}roblem.
\begin{definition}{ \normalfont \textbf{(Discrete Distance Problem $\DDP$).}}
    Given an initial configuration $\vec{s}\in\mathbb{F}_p^n$ and a {\it fixed} configuration 
    $\vec{t}\in\mathbb{F}_p^n$ of a linear cyclic hybrid cellular automata over the field $\F_p$
    $\LCCA$, find the discrete distance $\tau$ such that $\vec{t}=M^\tau \vec{s}$. 
\end{definition}

For a maximum-length $\LCCA$, such a $\tau$ is unique modulo $p^n-1$. However, for a group $\LCCA$ 
if $\vec{t}\notin \{M\vec{s}, M^2\vec{s}, \ldots, M^{d-1}\vec{s}\}$ where $M^d=I$ the identity matrix
then $\vec{t}$ is not reachable from $\vec{s}$ at all. In this case, $\tau=\infty$.

\subsection{\textsf{DDP} and \textsf{DLP} are Equivalent}
We start with a different but relevant issue of minimal polynomial $\mu_M(x)$ of the transition matrix $M$ in order to 
show that $M$ is always diagonalizable. Finding $\mu_M(x)$ needs no effort since  
the transition matrix $M$ of any linear finite state machine is non-derogatory~\cite{Cattell1996Analysis}.
For non-derogatory matrices the characteristic polynomial $f_M(x)=\mu_M(x)$.
As a linear $\CA$ is a special case of linear finite state machines the minimal polynomial 
of $\mu_M(x)=f_M(x)$.
However, we can show this explicitly as the following lemma.
\begin{lemma}\label{equal}
 For every $\LCCA$, the minimal polynomial $\mu_M(x)$ and the characteristic polynomial $f_M(x)$ of its 
 transition matrix $M$ are equal. 
\end{lemma}
\begin{proof}
 For every $\LCCA$, the characteristic polynomial is an irreducible polynomial of degree $n$ 
 over the field $\mathbb{F}_p$. Suppose, $\{\alpha,\alpha^p,\ldots,\alpha^{p^{n-1}}\}$ is the normal basis 
 for the extended field $\mathbb{F}_{p^n}$. Therefore, the minimal polynomial can be factorized  
 as $\mu_M(x)=\prod_{i=0}^{n-1}(x-\alpha^{p^i})=(x-\alpha)(x-\alpha^p)\ldots(x-\alpha^{p^{n-1}})$
 where $\alpha$ is a generator of the field $\mathbb{F}_{p^n}$.
 So, the degree of $\mu_M(x)=n$. We know that minimal polynomial divides the characteristic polynomial i.e., 
 $\mu_M(x)\mid f_M(x)$~\cite{Gallian2016Algebra}.
 So, these two following conditions need to be satisfied simultaneously, 
 \begin{enumerate}
  \item $\mathrm{deg}(\mu_M(x))=\mathrm{deg}(f_M(x))$ where deg$(\cdot)$ denotes the degree of a polynomial.
  \item $\mu_M(x)\mid f_M(x)$ where both $\mu_M(x)$ and $f_M(x)$ are monic irreducible monomials.
  \end{enumerate}
Both of these conditions together implies that $\mu_M(x)=f_M(x)$.
\end{proof}

Now we show that $M$ is always diagonalizable.
\begin{theorem}\label{diagonalizable}
 Transition matrix $M$ of any $\LCCA$ is diagonalizable over the field $\mathbb{F}_{p^n}/\langle \mu_M(x)\rangle$ 
 where $\mu_M(x)$ is the minimal polynomial of $M$ over the polynomial ring $\mathbb{F}_p[x]$.
\end{theorem}
\begin{proof}
 For any $\LCCA$, the transition matrix $M$ is in $\F_p^{n \times n}$, so its minimal polynomial $\mu_M(x)$ is 
 defined over the polynomial ring $\mathbb{F}_p[x]$. 
 Using Lemma~\ref{equal}, the characteristic polynomial $f_M(x)=\mu_M(x)=(x-\alpha)(x-\alpha^p)\ldots(x-\alpha^{p^{n-1}})$.
 By Knrocker's Theorem, $\F_{p^n}[x]/\langle f_M(x)\rangle$ is also a field having $f_M(x)=0$. 
 Therefore, $f_M(x)$ has $n$ distinct roots $\{\alpha,\alpha^p,\ldots,\alpha^{p^{n-1}}\}$ over the field 
 $\F_{p^n}/\langle f_M(x) \rangle$. 
 The eigenvalues of a matrix are the roots of its minimal polynomial.
 As $f_M(x)=\mu_M(x)$, $M$ has $n$ distinct eigenvalues $\{\alpha,\alpha^p,\ldots,\alpha^{p^{n-1}}\}$ over the field 
 $\F_{p^n}/\langle \mu_M(x)\rangle$ . 
 
 Since a matrix is diagonalizable over the field $\mathbb{F}_{p^n}$ if and only if the factors of its minimal polynomial 
 has $n$ distinct roots over $\F_{p^n}$, $M$ is diagonalizable over 
 the field $\F_{p^n}/\langle \mu_M(x)\rangle$.
\end{proof}

% In order to reduce a problem $A$ into another problem $B$ in polynomial time, we will assume that there is 
% an oracle $\mathcal{B}$ that solves instances of problem $B$. We will encode the inputs of problem $A$ 
% to the inputs of problem $B$ and will call the oracle $\mathcal{B}$ for polynomial time in the input length $n$ 
% to solve problem $A$.
\subsection{The Reduction DDP$\le_{p}$DLP}
\begin{theorem}\label{thm:forward}
 If there exists an algorithm $\D$ that efficiently solves the problem $\DLP$ over the field $\F_{p^n}$ then 
 there exists an algorithm $\d$ that efficiently solves the problem of $\DDP$.
\end{theorem}

\begin{proof}
$\D$ takes a generator $g\in\mathbb{F}_{p^n}^*$ and an element $a\in\mathbb{F}_{p^n}^*$ as its 
input and outputs an $x$ such that $a=g^x$. 
% We will see that either $d|n$ or $d=n$ in Sect.~\ref{relation}. 
The inputs of a $\DDP$ instance are
% the transition matrix $M\in\mathbb{F}^{n\times n}_p$, 
two configurations $\vec{s},\vec{t}\in\mathbb{F}^n_p$.
% and a final configuration $\vec{t}\in\mathbb{F}^{n}_p$. 
$\d$ needs to output a $\tau\le  p^n-1$
such that $\vec{t}=M^\tau \vec{s}$. 

Using lemma~\ref{diagonalizable}, $\d$ first diagonalizes $M=Q \Lambda Q^{-1}$ uniquely, such that,
\begin{align*}
M &=
\overbrace{
\begin{pmatrix}
q_{11} & q_{12} & \cdots & q_{1n} \\
q_{21} & q_{22} & \cdots & q_{2n} \\
\vdots & \vdots & \ddots & \vdots \\
q_{n1} & q_{n2} & \cdots & q_{nn} \\
\end{pmatrix}}^{Q}  
\overbrace{
\begin{pmatrix}
\alpha & 0 & \cdots & 0 \\
0 & \alpha^p & \cdots & 0 \\
\vdots & \vdots & \ddots & \vdots   \\
0 & 0 & \cdots & \alpha^{p^{n-1}} \\
\end{pmatrix}}^{\Lambda} 
\overbrace{
\begin{pmatrix}
q'_{11} & q'_{12} & \cdots & q'_{1n} \\
q'_{21} & q'_{22} & \cdots & q'_{2n} \\
\vdots & \vdots & \ddots & \vdots \\
q'_{n1} & q'_{n2} & \cdots & q'_{nn} \\
\end{pmatrix}}^{Q^{-1}}\\
M^\tau &=
\overbrace{
\begin{pmatrix}
q_{11} & q_{12} & \cdots & q_{1n} \\
q_{21} & q_{22} & \cdots & q_{2n} \\
\vdots & \vdots & \ddots & \vdots \\
q_{n1} & q_{n2} & \cdots & q_{nn} \\
\end{pmatrix}}^{Q}  
\overbrace{
\begin{pmatrix}
(\alpha^\tau) & 0 & \cdots & 0 \\
0 & (\alpha^\tau)^{p} & \cdots & 0 \\
\vdots & \vdots & \ddots & \vdots   \\
0 & 0 & \cdots & (\alpha^\tau)^{p^{n-1}} \\
\end{pmatrix}}^{\Lambda^\tau} 
\overbrace{
\begin{pmatrix}
q'_{11} & q'_{12} & \cdots & q'_{1n} \\
q'_{21} & q'_{22} & \cdots & q'_{2n} \\
\vdots & \vdots & \ddots & \vdots \\
q'_{n1} & q'_{n2} & \cdots & q'_{nn} \\
\end{pmatrix}}^{Q^{-1}}
\end{align*}
Therefore, assuming $x[i]$ denotes the $i$-th coordinate of any vector $\vec{x}$.
\begin{align}\label{decomposition}
 \begin{split}
 \vec{t} &=M\vec{s}\\
\overbrace{\begin{pmatrix}
 t[1]\\
 t[2]\\
 \vdots\\
 t[n]
\end{pmatrix}}^{\vec{t}} &=
\overbrace{
\begin{pmatrix}
q_{11} & q_{12} & \cdots & q_{1n} \\
q_{21} & q_{22} & \cdots & q_{2n} \\
\vdots & \vdots & \ddots & \vdots \\
q_{n1} & q_{n2} & \cdots & q_{nn} \\
\end{pmatrix}}^{Q}  
\overbrace{
\begin{pmatrix}
(\alpha^\tau) & 0 & \cdots & 0 \\
0 & (\alpha^\tau)^{p} & \cdots & 0 \\
\vdots & \vdots & \ddots & \vdots   \\
0 & 0 & \cdots & (\alpha^\tau)^{p^{n-1}} \\
\end{pmatrix}}^{\Lambda^\tau} 
\overbrace{
\begin{pmatrix}
q'_{11} & q'_{12} & \cdots & q'_{1n} \\
q'_{21} & q'_{22} & \cdots & q'_{2n} \\
\vdots & \vdots & \ddots & \vdots \\
q'_{n1} & q'_{n2} & \cdots & q'_{nn} \\
\end{pmatrix}}^{Q^{-1}}
\overbrace{\begin{pmatrix}
 s[1]\\
 s[2]\\
 \vdots\\
 s[n]
\end{pmatrix}}^{\vec{s}} \\
\overbrace{\begin{pmatrix}
 t[1]\\
 t[2]\\
 \vdots\\
 t[n]
\end{pmatrix}}^{\vec{t}} 
&=
 \begin{pmatrix}
(q_{11}s[1]q'_{11} +\cdots + q_{1n}s[n]q'_{n1}) & \cdots & (q_{11}s[1]q'_{1n} +\cdots + q_{1n}s[n]q'_{nn}) \\
(q_{21}s[1]q'_{11} +\cdots + q_{2n}s[n]q'_{n1}) & \cdots & (q_{21}s[1]q'_{2n} +\cdots + q_{2n}s[n]q'_{nn}) \\
 \vdots  & \vdots  & \vdots\\
(q_{n1}s[n]q'_{n1} +\cdots + q_{nn}s[n]q'_{n1}) & \cdots & (q_{n1}s[1]q'_{nn} +\cdots + q_{nn}s[n]q'_{nn})  \\
\end{pmatrix}
\begin{pmatrix}
 \alpha^\tau\\
 (\alpha^p)^\tau\\
 \vdots\\
 (\alpha^{p^{n-1}})^\tau
\end{pmatrix}
 \end{split}
\end{align}

Given $\vec{t}$ and $\vec{s}$, $\d$ finds all the the $\alpha^{\tau p^i}$ by Gaussian elimination. 
Then $\d$ passes the pair $(\alpha, \alpha^\tau)$ to $\D$. When $\D$ outputs $\tau$, $\d$ outputs $\tau$.

\end{proof}

\subsection{The Reduction DLP$\le_{p}$DDP}
This reduction is much more straight forward than the previous one. 
% In this case we assume that there is an oracle $\d$ that efficiently solves \textsf{DDP} with the inputs 
% $\vec{s}\in\F_p^n$, $\vec{t}\in\F_p^n$ and $M\in\mathbb{F}_{p}^{n \times n}$. 
% Given the input pair an element $a\in\mathbb{F}_{p^n}$ and a generator $g\in\mathbb{F}_{p^n}$
% $\DLP$ asks to find the $x$ such that $a=g^x$. We will find the $x$ using $\d$.
The main difference between this reduction with the previous one is that in this case given the generator $g\in\F_{p^n}$ 
for $\DLP$ we need to find an $M=Q \Lambda Q^{-1}$ for $\DDP$ such that the diagonal entries of the diagonal matrix $\Lambda_{ii}$ 
become $g^{p^i}$.  We do it by the following lemma.

\begin{lemma}\label{transformation}
The left multiplication by any element $a\in\mathbb{F}_{p^n}$ translates to 
a linear transformation $T_a$ over the field $\mathbb{F}_{p^n}$.
\end{lemma}

\begin{proof}
Suppose $\mathbb{B}=\{\beta_0,\beta_1,\beta_2,\ldots,$ $\beta_{n-1}\}$ form an
$\F_p$-basis of $\F_{p^n}$. For example, we can represent 
$\F_{p^n}=\F_p[x]/\langle f(x)\rangle$, and consider the polynomial basis
$1,x,x^2,\ldots,x^{n-1}$. Every $\beta_i\beta_j$ can be written
as an $\F_p$-linear combination of the basis elements.

Fix an element $a=\alpha_0\beta_0+\alpha_1\beta_1+\alpha_2\beta_2+\cdots+\alpha_{n-1}\beta_{n-1}\in\F_{p^n}^*$, 
where, $\alpha_i\in\mathbb{F}_{p^n}$. Now, take an arbitrary element
% $\gamma=c_0\beta_0 + c_1\beta_1 + c_2\beta_2 + \cdots + c_{n-1}\beta_{n-1}=\beta_0 & \beta_1 & \beta_2 &\cdots & \beta_{n-1}$
\begin{equation}
 b=\gamma_0\beta_0 + \gamma_1\beta_1 + \gamma_2\beta_2 + \cdots + \gamma_{n-1}\beta_{n-1} = 
  (\beta_0 \ \beta_1 \ \beta_2 \cdots  \beta_{n-1})
 \begin{pmatrix} \gamma_0\\ \gamma_1\\ \gamma_2\\ \vdots \\ \gamma_{n-1} \end{pmatrix} 
\end{equation}
assuming $\gamma_i\in\F_{p^n}$.
Expand the product $ab$, and replace each $\beta_i\beta_j$
as the (known) linear combination of the basis elements. This lets
us write
\begin{equation}
 ab=(\beta_0 \ \beta_1 \ \beta_2 \cdots  \beta_{n-1}) T_{\alpha} \begin{pmatrix} \gamma_0\\ \gamma_1\\ \gamma_2\\ \vdots \\ \gamma_{n-1} \end{pmatrix},
\end{equation}
where $T_a\in \F_p^{n \times n}$. That is, multiplication by $a$ can be viewed as the linear transformation $T_a$.

Conversely, given a matrix $T_a$ can be 
obtained efficiently by solving the equation $det(xI-T_a)$. In particular, $a$ is one of the 
eigenvalues of $T_a$.
\end{proof}

\begin{theorem}\label{thm:backward}
 If there is an algorithm $\d$ that efficiently solves the problem $\DDP$ then there is an efficient algorithm 
  $\D$ that solves the problem of $\DLP$.
\end{theorem}

\begin{proof}
 Given two arbitrary configurations $\vec{s},\vec{t}\in\mathbb{F}_p^n$ of an $\LCCA$, 
 $\d$ outputs a $\tau$ such that $\vec{t}=M\vec{s}$.
 Given an instance $(g,a)$ of $\DLP$, $\D$ finds the linear transformations $T_g$ and $T_a$ using Lemma~\ref{transformation}.
$\D$ takes an arbitrary nonzero $\vec{s}\in\mathbb{F}_p^n$ and obtains $\vec{t}=T_a \vec{s}$. 
$\D$ feeds $\vec{s}$, $\vec{t}$ and $T_g$ to $\d$ to have the solution $\tau$ as $\vec{t}= T_a \vec{s}=(T_g)^\tau \vec{s}$.
\end{proof}

We did not construct $M=Q \Lambda Q^{-1}$ where $\Lambda_{ii}=g^{p^i}$ because $Q_{ij}=f(g)_{i}^{q^j}$. So the cost of 
reduction would have been more. Similarly, linear transformations were not used for 
the previous reduction as given a $\vec{t}$ and $\vec{s}$, to find the $M^\tau$ we required 
to solve the system of matrix equations $M^i \vec{t} = M^\tau M^i \vec{s}$, but spending more time.  

\subsection{Explaining The Equivalence}\label{relation}

Here we show that $\mathbb{F}_p\langle\alpha \rangle$ and $\mathbb{F}_p\langle M \rangle$ are isomorphic and this is 
the key reason that $\DLP$ and $\DDP$ equivalent.

As already mentioned that the characteristic polynomial $f_M(x)$ of any $\LCCA$ is irreducible.
Since $f_M(x)$ is irreducible polynomial of degree $n$, the polynomial ring $\F_p[x]/\langle f_M(x) \rangle $ 
must be isomorphic to the field $\mathbb{F}_{p^n}$.
Suppose $\alpha$ is one of the roots of $f_M(x)$ over the field $\F_{p^n}$ such that $M=T_\alpha$, then
the field $\F_p\langle\alpha\rangle=\{0,1,\alpha,\ldots,\alpha^{p^d-2}\}$ 
must be a subfield of $\F_p[x]/ \langle f_M(x) \rangle$ where $d \mid n$. When $d=n$ and $\alpha^i\ne 1$ for all $i<p^n-1$ then 
$\alpha$ is called a primitive element of the field $\F_{p^n}$ and $f_M(x)$ is called a primitive polynomial. 
As $f_M(x)$ is the characteristic polynomial of the matrix $M$ by Caley-Hamilton theorem $f(M)=\mathbf{0}$ over 
the field $\F_p$. And using Lemma~\ref{transformation}, we can replace $\alpha$ with the matrix $M=T_\alpha$. 
Therefore, $\mathbb{F}_p\langle\alpha\rangle\simeq\mathbb{F}_p\langle M\rangle=\{\mathbf{0},I,M,\ldots,M^{p^d-2}\}$. 
In fact, using the matrices of $\mathbb{F}_p\langle M\rangle$ and a state $\vec{s}_i \in\mathbb{F}_p^n$ 
we can have a set of disjoint sequences $S_i=\langle \mathbf{0},I.\vec{s}_i,M.\vec{s}_i,\ldots,M^{p^d-2}.\vec{s}_i\rangle$.
We have $\bigcap\limits_{i=1}^{\frac{p^n-1}{p^d-1}} S_{i}=\emptyset$ and 
$\bigcup\limits_{i=1}^{\frac{p^n-1}{p^d-1}} S_{i}=\mathbb{F}_p^n$. Essentially, the sequences $S_i$ are the cycles 
of a linear cyclic hybrid group cellular automata. When $f_M(x)$ is a primitive polynomial, $M$ becomes a primitive element 
of the field $\F_p\langle M\rangle$ yielding a maximum length linear cyclic hybrid cellular automata.

On the other hand, as $f_M(M)=\mathbf{0}$, we define a subfield 
$\F\langle M \rangle=\{\mathbf{0},I,M,\ldots,M^{2^d-2}\}$ 
of $\F_p[x]/ \langle f_M(x) \rangle$.
In particular, when $f_M(x)$ is a primitive polynomial the primitive roots $\alpha$ and $M$ generate the entire 
field $\mathbb{F}_p[x]/ $ as $d=n$. Otherwise, when $d \mid n$, $\F \langle \alpha \rangle$ and 
$\F \langle M \rangle$
are subfields of $\mathbb{F}_p[x]/\langle f_M(x) \rangle$. 
We know that two fields of the same characteristics and same order are 
isomorphic. Hence, $\F \langle \alpha \rangle$ and $\F\langle M \rangle$ are isomorphic.
Any isomorphism between two finite fields is efficient to compute~\cite{Lenstra1991Isomorphism}.
Therefore, if one among the problems $\DDP$ and $\DLP$ was easier than the other one, then 
we could have always turned down the harder problem into the easier one using this isomorphism. So, 
the equivalence between these two could not be proven.

\section{Two Important Variants of DDP}

\begin{definition}{ \normalfont \textbf{( Fixed (coordinates) \textsf{DDP} $\FDP$).}}
For a fixed $k<n$, given a configuration $\vec{s}\in\F_p^n$ of an $n$-cell 
$\LCCA$, and a vector $\vec{x} \in \F_p^k$,
the problem of $\FDP$ is to find a $\tau$ such that 
$M^\tau \vec{s} \in  \{(\vec{x}|\vec{a}) \mid \forall \vec{a} \in \F_p^{n-k}\}$. 
\end{definition}

By the vector $(\vec{x}|\vec{a})$ we mean the augmentation of two vectors $\vec{x}$ and $\vec{a}$.
So $\{(\vec{x}|\vec{a}) \mid \forall \vec{a} \in \F_p^{n-k}\}$ is the set of 
vectors whose $k$ coordinates are fixed to $\vec{x}$. As the coordinates are independent of 
each other so the problem $\FDP$ also remains independent of the choice of the $k$ coordinates and 
the choice for $\vec{x}$ too.

\begin{description}\item [The Reduction $\FDP\le _{p}\DDP$] 
This reduction is straightforward as we are free to choose any 
$\vec{t}\in \{(\vec{x}|\vec{a}) \mid \forall \vec{a} \in \F_p^{n-k}\}$ and call the 
$\DDP$ oracle $\d$ on the inputs $\vec{t}$ and $\vec{s}$. The reply $\tau$ will be a valid solution of 
the $\FDP$ instance as $M^\tau \vec{s}=\vec{t}$.

\item [Does $\DDP\le _{p}\FDP$ ?] Presently, we have no answer for this question 
except the belief that the answer is a no. We reason this "no" with the following  
difference between $\FDP$ and $\DDP$. 
% Let us reformulate both these problems in terms of basis.
The solution of $\DDP$ is unique modulo $p^n-1$ while there exist $p^{n-k}-1$ solutions for $\FDP$.
This lack of uniqueness of solutions makes $\FDP$ easier than $\DDP$. Thus this reduction should not be possible.
\end{description}

\begin{definition}{ \normalfont  \textbf{( Short (Fixed) \textsf{DDP} $\SDP$).}}
For a fixed $k<n$ and a $\delta\in\Z^+$,  
given a configuration $\vec{s}\in\F_p^n$ of an $n$-cell $\LCCA$, and 
a vector $\vec{x} \in \F_p^k$,
the problem of $\SDP$ is to find a $\tau<\delta$ such that 
$M^\tau \vec{s} \in  \{(\vec{x}|\vec{a}) \mid \forall \vec{a} \in \F_p^{n-k}\}$. 
\end{definition}

By the above argument $\SDP$ also remains independent of the choice for $k$ coordinates and the vector $\vec{x}$.

Due to the restriction $\tau<\delta$, unlike $\DDP$, we have no hint if $\SDP$ reduces to 
$\DDP$, if at all. A famous problem that becomes hard due to the presence of such restriction on 
the norm-bound of its solution is short integer solution problem $\SIS$ ~\cite{Ajtai1996SIS}. 
% We will see $\SIS$ reduces to $\SDP$ in the following section. 
% However, we guess that Ring-$\SIS$ suits more this reduction.

Our goal is to determine the hardness of the problem $\SDP$. It appears that 
the reduction $\SIS\le _{p} \SDP$ is not immediate. Solutions of $\SIS$ 
are integer vectors while that of $\SDP$ are vectors over $\F_q$. 
Therefore, we define another problem, namely Short Lacunary Solution Problem
($\SRPS$), as the missing link between $\SDP$ and $\SIS$. In particular, 
we will show that, $\SIS \le _{p} \SRPS$ followed by $\SRPS \le _{p} \SDP$. 
As Turing reductions are transitive, we will conclude that $\SIS \le _{p} \SDP$.

% 
% \begin{definition} {\normalfont  \textbf {(Efficient Order Isomorphism).}}\quad
% We call $\sigma:\F_q \rightarrow \Z_q$ to be an efficient order isomorphism if
% \begin{enumerate}
%  \item it is an efficiently computable and invertible.
%  \item $\sigma(\alpha^i)>\sigma(\alpha^j)$ if only if $i>j$.
%  \item $\sigma(0)=0$. 
% \end{enumerate}
% \end{definition}

\bibliographystyle{splncs04}

\bibliography{ref}
\end{document}